\documentclass[conference]{IEEEtran}

\usepackage{amsmath}    
\usepackage{graphics,epsfig,subfigure}    
\usepackage{verbatim}   
\usepackage{color}      
\usepackage{subfigure}  
\usepackage{amssymb}
\usepackage{wasysym}
\usepackage{epstopdf}
\usepackage{graphicx}
\usepackage{mathrsfs}

\newcommand{\expect}[1]{\mathbb{E}\big\{#1\big\}}

\newcommand{\script}[1]{{{\cal{#1} }}}

\newtheorem{coro}{\textbf{Corollary}}
\newtheorem{lemma}{\textbf{Lemma}}

\newtheorem{theorem}{\textbf{Theorem}}

\allowdisplaybreaks
\begin{document}

\title{Optimizing Your Online-Advertisement Asynchronously}
\author{\large{Longbo Huang} \\
 longbohuang@tsinghua.edu.cn\\
 IIIS, 
 Tsinghua University
 \thanks{Longbo Huang  (http://www.iiis.tsinghua.edu.cn/$\sim$huang) is with the Institute for Theoretical Computer Science and the Institute for Interdisciplinary Information Sciences, Tsinghua University, Beijing, P. R. China. 
} 
\thanks{This work was supported in part by the National Basic Research Program of China Grant 2011CBA00300, 2011CBA00301, the National Natural Science Foundation of China Grant 61033001, 61061130540, 61073174. }
}

\maketitle
\thispagestyle{empty}
\pagestyle{empty}

\begin{abstract}
We consider the problem of designing optimal online-ad investment strategies for a single advertiser,  who invests at multiple sponsored search sites simultaneously, with the objective of maximizing his average revenue subject to the advertising budget constraint.  A greedy online investment scheme is developed to achieve an average revenue that can be pushed to within $O(\epsilon)$ of the optimal, for any $\epsilon>0$, with a tradeoff that the temporal budget violation is $O(1/\epsilon)$. Different from many existing algorithms, our scheme allows the advertiser to \emph{asynchronously} update his investments on each search engine site, hence applies to systems where the timescales of action update intervals are heterogeneous for different sites. We also  quantify the impact of inaccurate estimation of the system dynamics and show that the algorithm is robust against imperfect system knowledge. 
\end{abstract}

\section{Introduction} 
Over the past years, online advertisement has  grown significantly and has become one of the most important advertising approaches. According to Emarketer analysts  \cite{emarketer-2012}, online advertising spending will reach $46.5$ billion dollars in $2013$ and is expected to surpass $60$ billion dollars in $2016$. Among the many online advertising methods, sponsored search has been a very successful mechanism and has served as a major income source for many search engines such as Google and Yahoo!. 

%

Due to the popularity of sponsored search, designing efficient algorithms for it has received much attention, and there has been a large body of previous work in this area. These works can roughly be divided into two categories. The first category of work aims at designing optimal algorithms for maximizing the search engines' revenue. \cite{adword-online-2007} proposes an online ad allocation algorithm based on online matching \cite{karp-online-bm-1990}, and proves that it achieves an $1-1/e$ competitive ratio. \cite{niv-ad-2007} uses a primal-dual approach and designs a scheme that achieves the same competitive ratio. 
\cite{ranking-auction-2007} looks at the problem of finding revenue-optimal rules for ranking ads. \cite{online-ad-opt-2011} formulates the problem as a stochastic optimization problem and proposes an algorithm to achieve near-optimal performance. \cite{ad-budget-2013} designs algorithms for associating queries with budget constrained advertisers and conducts system implementations. 
The second category focuses on designing optimal bidding strategies for the advertisers. \cite{bb-bidding-2007} studies  properties of greedy bidding schemes and develops a balanced bidding algorithm. \cite{amin-2012} considers a single advertiser trying to maximize the number of clicks given a certain budget and tackles the problem with Markov decision process (MDP).  \cite{opt-bidding-2012} formulates the optimal bidding problem as an MDP with large number of bidders. \cite{online-ad-bidding-2013} considers a single advertiser and proposes a greedy bidding scheme based on stochastic approximation. 
However, we note that the aforementioned works only consider optimizing the performance of the advertiser or the search engine in the ad auction process. Hence, they focus mainly on designing algorithms for systems with a single search engine site. In practice, however, an advertiser can typically utilize multiple sites simultaneously for advertising. In order to maximize his revenue, one must jointly optimize the investments at all sites. 

In this paper, we consider a single advertiser who is trying to advertise his product by utilizing a set of  sponsored search engines (called sites below), e.g., Google or Yahoo!. At every time, the advertiser decides the advertising configuration at each site, e.g., total budget and maximum pay-per-click, associating ads to different keywords, and the bidding strategies, and decides how much money to allocate to advertising at each search engine. According to the advertiser's customized setting and the sites' ad allocation and billing mechanisms, the amount of money at a search engine site will be consumed at a certain rate, and an amount of revenue will be generated to the advertiser. After the money is completely depleted at a site, the advertiser updates his investment at the site and resets his advertising  configurations, based on his investment strategy and his observation of the site's performance. 
%
The objective of the advertiser is to find an investment policy to maximize his time average revenue subject to his advertising budget. 

This problem captures many features of the online advertising procedure at popular sponsored search engine sites such as Google's AdWords \cite{google-adword} and Bing ads \cite{bing-ad}. 
However,  solving this problem is challenging. First of all, the advertiser has to learn over time about how to allocate his investment across different search engine sites subject to his budget constraint. With the many different  features and ad mechanisms of the different sites and the many options of the advertiser, this is a nontrivial task.  Second, due to the intrinsic stochastic nature of the online advertising process, the advertising performance as well as the time to re-update the advertising configurations at each site are highly dynamic and can be very different from site to site. This imposes the constraint that the investment strategy must be able to handle \emph{asynchronous} system operations. However, most existing optimal control algorithms, e.g., \cite{neelynowbook} \cite{huang-lifo-ton},  require simultaneous update of the control actions and hence do not apply in this case. Finally, optimizing  time average metrics defined over time-varying update intervals  typically involves optimizing sum of ratios of functions. This is in general non-convex and problems of this kind are generally very hard to solve. 

In order to tackle this problem and to resolve the aforementioned difficulties, we adopt the recently developed Lyapunov technique for renewal systems and asynchronous control \cite{neelymcbook}  \cite{neely-asyn}. This approach has three important components: (1) a virtual deficit queue with carefully defined arrival and service rates as ratios of the corresponding system metrics, (2) a Lyapunov drift defined over variable size time intervals, and (3) a drift-plus-penalty-ratio principle in decision making. 
These three components enable the development of a simple yet near-optimal asynchronous control scheme called \textsf{Ad-Investment (AI)}, which allows asynchronous updates of the investments and the advertising configurations. We show that the \textsf{AI} algorithm achieves an average revenue that can be pushed to within $O(\epsilon)$ of the maximum, for any $\epsilon>0$, with a tradeoff of an $O(1/\epsilon)$ temporal budget violation. We further investigate the performance of \textsf{AI} with imperfect system knowledge  and quantify its impact on the algorithm performance. 

This paper is organized as follows: In Section \ref{section:model}, we state our model and problem formulation. In Section \ref{section:revenue}, we derive a characterization of the maximum revenue. Then, we present the dynamic \textsf{Ad-Investment} strategy  (\textsf{AI}) for the advertiser in Section \ref{section:ai}. We investigate the performance loss of \textsf{AI} due to imperfect knowledge of the system dynamics in Section \ref{section:approx}. Simulation results are presented in Section \ref{section:sim}. We conclude the paper in Section \ref{section:conclusion}. 

\vspace{-.02in}
\section{System Model}\label{section:model}
In this section, we present our system model. We consider a single advertiser who is trying to advertise his products at $N$ sponsored search websites (called sites below), e.g., Google or Bing, denoted by $\script{S}=\{S_1, ..., S_N\}$ (see Fig. \ref{fig:system}). We assume that the system operates in continuous time, i.e., $t\in\mathbb{R}_+$. 
\begin{figure}[cht]
\centering
\vspace{-.08in}
\includegraphics[height=1in, width=2.8in]{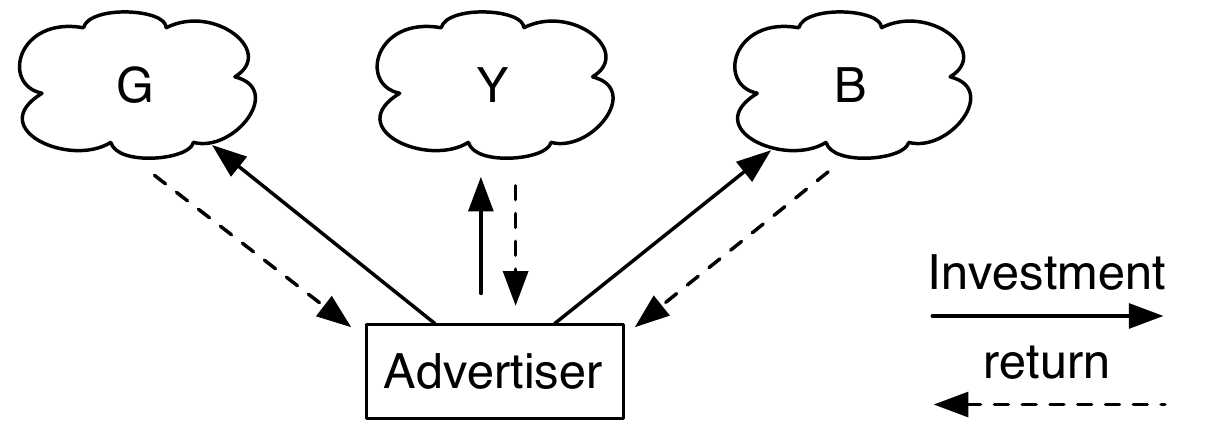}
\vspace{-.05in}
\caption{A single advertiser allocates his advertising budget across multiple search engine sites (``G,'' ``Y'' and ``B'' in this example). The objective of the advertiser is to achieve maximum revenue subject to his advertising budget.}\label{fig:system}
\vspace{-.05in}
\end{figure}

\vspace{-.11in}
\subsection{Advertiser Actions}
The advertiser first decides how much money to invest at each site and specifies the operation configuration, e.g., the maximum pay-per-click, or which ads to play for the keywords. Then, when the money deposited at a site is depleted, the advertiser  updates his investment amount and the operation configuration for that  search site, and the advertising process continues. \footnote{Note that our approach can also handle the  more general case where the advertiser uses other criteria for deciding when to update the investment amounts and the advertising configurations. } 

%
Note that this operating mode is quite common in today's sponsored search engines, e.g., Google Adwords \cite{google-adword}.  Under the advertiser's investment policy, each site goes through a sequence of asynchronous \emph{advertising} (and \emph{freeze}) intervals, defined as  the intervals during which the advertiser has nonzero (zero) budget left at that site (Fig. \ref{fig:service-idle-frame} shows an example). In the following, we will conveniently call an advertising interval plus the subsequent freeze interval a \emph{frame}. 
\begin{figure}[cht]
\centering
\vspace{-.06in}
\includegraphics[height=1.8in, width=3.3in]{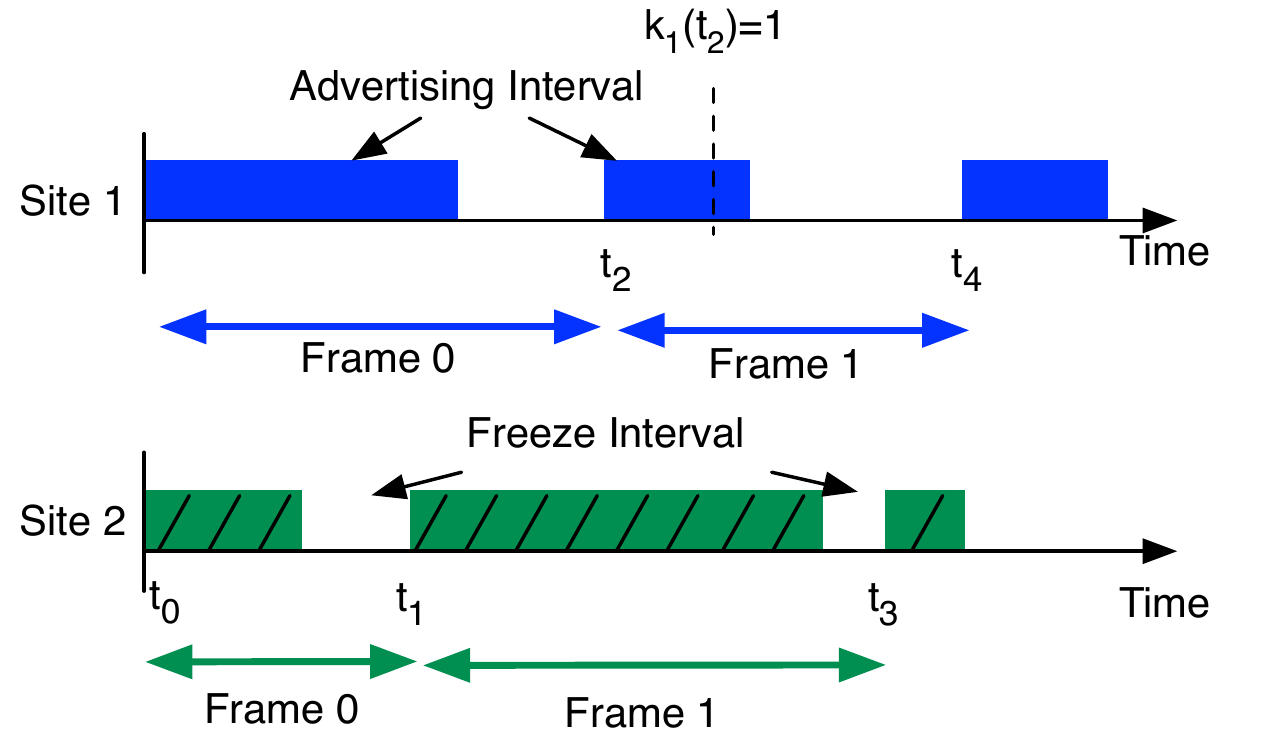}
\vspace{-.06in}
\caption{An example of the advertising and freeze intervals for a 2-site system. The advertiser updates the configurations of the two sites asynchronously at times $t_0$, $t_1$, $t_2$, $t_3$ and $t_4$. }\label{fig:service-idle-frame}
\vspace{-.05in}
\end{figure}

We use $p_n[k]$ to denote the amount of new investment the advertiser makes to search site $S_n$ for the $k^{\text{th}}$ advertising interval.  Note that this new  investment will start a new advertising interval. To allow the advertiser to temporarily suspend his investment at a site, we  assume that the advertiser can always adopt a ``freeze'' option. Under this freeze option, the advertiser will not make further new investment into the search site for some time $T^{\text{fr}}_{n}[k]$ after the advertising interval. We call this period the freeze interval. 
%
We assume that for all $k$, the pair $(p_n[k], T^{\text{fr}}_{n}[k])$ is chosen from the feasible (investment, freeze time) set $\script{P}_n$ for all $n$. 
We assume that $\script{P}_n$ includes the option $p_n[k]=0$, and the constraints $0\leq p_n[k]\leq p_n^{\max}$  and $0\leq T^{\text{fr}}_{n}[k]\leq T^{\text{fr},\max}_{n}$. 
One such example can be $\script{P}_n=[0, p]\times[0, T]$.  
We also  use $m_n[k]$ to denote the operation configuration, for instance, the maximum pay-per-click together with the maximum per-day payment, or the targeted user groups, or what ads to display for the corresponding keywords, or the bidding strategy to use at that site for that frame. We assume that $m_n[k]\in\script{M}_n$, where $\script{M}_n$  is the set of feasible operation configurations at site $S_n$. We assume that both $\script{P}_n$ and $\script{M}_n$ are compact sets. 
 
Given the investment amount and the operating configuration, the corresponding search sites will use their own pricing and ad allocation mechanisms to allocate advertisement opportunities to the advertisers, e.g., the generalized second price auction \cite{gsp2007} or the VCG mechanism \cite{agt-book}, and to generate revenue for the advertiser, e.g., by bringing visiting traffic or by directly bringing product orders to the advertiser, until the advertiser's new investment is fully depleted. However, due to the dynamic nature of the advertising process,  the outcome at each site will inevitably be time-varying, e.g., the volume of the visiting traffic to the search site and the set of keywords that are queried. Such a dynamic condition will affect the site's advertising performance,  resulting in different amounts of generated profit to the advertiser and  time-varying speeds in depleting the remaining advertisement deposit at the site. 
To capture this dynamics, we denote the duration taken to deplete the $k^{\text{th}}$ investment at site $S_n$ by $T^{\text{ad}}_{n}[k]$, and assume that $T^{\text{ad}}_{n}[k]$ is an independent random variable conditioning on the advertiser's investment $p_n[k]$ and the operating configuration $m_n[k]$. We further assume that there exists a function $F_{n}[k]\triangleq F_{n}(p_n[k], m_n[k])$, which maps $(p_n[k], m_n[k])$ to the expected length of $T^{\text{ad}}_{n}[k]$, i.e., 
\begin{eqnarray} 
F_{n}: (p_n[k], m_n[k]) \mapsto \expect{T^{\text{ad}}_{n}[k]}. \label{eq:ad-time-def}
\end{eqnarray} 
The function $F_{n}(\cdot, \cdot)$ is assumed to be known to the advertiser. We assume that under any feasible actions, the frame size at each site $S_n$ is both upper and lower bounded, i.e., \footnote{Note that this is a reasonable assumption. If the advertiser invests nonzero money (typically must be more than a minimum), it will take some nonzero time to deplete the money. Otherwise if the advertiser suspends his investment at the site, then he will specify a nonzero freeze time and resume later. }
\begin{eqnarray}
0<T_{n}^{\min}\leq  T^{\text{ad}}_{n}[k] + T^{\text{fr}}_{n}[k]\leq T_{n}^{\max}. \label{eq:time-max}
\end{eqnarray}

We denote $R_n[k]$ the total revenue the advertiser earns during advertising interval $k$ from site $S_n$.  Similar to $T^{\text{ad}}_{n}[k]$, $R_n[k]$ is also a random variable which is conditionally independent of others given $p_n[k]$ and $m_n[k]$. We also assume that there exists a function $G_{n}[k]\triangleq G_{n}(p_n[k], m_n[k])$ that: 
\begin{eqnarray}
G_{n}: (p_n[k], m_n[k]) \mapsto \expect{R_n[k]}. \label{eq:ad-revenue-def}
\end{eqnarray} 
The function $G_n$ satisfies $G_n(0, \cdot)=0$, $G_{n}(\cdot, \cdot)\geq0$, $ G_n(p, m)\leq \nu p$ for some constant $\nu$, and $G_n(p, m)<G^{\max}$ for all $n$ and all $m, p$. We also assume that it is known to the advertiser. 

In the following of the paper, for notational convenience, we define common upper and lower bounds by dropping the index ``$n$,'' i.e., $T^{\max}=\max_nT_n^{\max}$, $T^{\min}=\min_nT_n^{\min}$,  $R^{\max}=\max_nR_n^{\max}$ and $p^{\max}=\max_np_n^{\max}$. 

\subsection{Budget Constraint and Objective} 
It can be seen from the above that the advertiser's time average revenue is given by:  \footnote{Throughout this paper, we assume that all limits and corresponding time average values exist with probability $1$. } 
\begin{eqnarray}
\text{Profit}_{\text{av}}\triangleq\sum_{n}\frac{\lim_{K\rightarrow\infty}\frac{1}{K}\sum_{k=0}^{K-1}R_n[k]  }{ \lim_{K\rightarrow\infty}\frac{1}{K}\sum_{k=0}^{K-1}(T^{\text{ad}}_n[k] + T^{\text{fr}}_n[k]  ) }.  \label{eq:revenue}
\end{eqnarray}
%
In practice, advertisers typically also have limited advertising budgets. We model this by requiring that the  time average advertising expenditure is no more than some pre-specified budget value $B_{\text{av}}$, i.e., 
\begin{eqnarray}
\sum_{n}\frac{\lim_{K\rightarrow\infty}\frac{1}{K}\sum_{k=0}^{K-1}p_n[k]  }{ \lim_{K\rightarrow\infty}\frac{1}{K}\sum_{k=0}^{K-1}(T^{\text{ad}}_n[k] + T^{\text{fr}}_n[k]  ) } \leq B_{\text{av}}. \label{eq:budget-cond}
\end{eqnarray}
Here the term on the left-hand-side (LHS) is the time average advertising expenditure of the advertiser. 

The advertiser's objective is to find an investment strategy that determines how to make the ad-investments, i.e., choosing $p_n[k]$, $T^{\text{fr}}_n[k]$ and $m_n[k]$, so as to maximize his time average revenue (\ref{eq:revenue}) subject to the average advertisement budget constraint  (\ref{eq:budget-cond}).  
Below, we call a strategy that chooses $(p_n[k], T^{\text{fr}}_n[k])\in\script{P}_n$ and $m_n[k]\in\script{M}_n$, and satisfies the budget constraint (\ref{eq:budget-cond}) a \emph{feasible} strategy. Then, we use $\text{Profit}_{\text{av}}^*$ to denote the maximum time average budget-constrained profit achievable over all feasible strategies. 

\vspace{-.03in}
\subsection{Discussion of the Model}
Our model is general and does not assume specific structures of the system.  The operation configuration $m_n[k]$ can be used to describe the specific bidding policy used during the advertising interval, i.e., automatic bidding or customized bidding scheme, e.g., \cite{bb-bidding-2007}. In this case, the functions $G_n(\cdot, \cdot)$ and $F_n(\cdot, \cdot)$ represent the expected revenue and interval length under such bidding schemes.  
Therefore, the strategies developed in this paper can indeed be implemented with the existing bidding schemes. 
Under such general settings, solving this problem is challenging, since the performance metrics in (\ref{eq:revenue}) and (\ref{eq:budget-cond}) are equivalent to sum of ratios of functions and are generally non-convex. Moreover, different search sites may need updates asynchronously, whereas most existing optimal control algorithms \cite{neelynowbook} \cite{huang-lifo-ton} require synchronous  action updates at all components of the system. 

\section{Characterizing Optimal Revenue}\label{section:revenue}
In this section, we first present a theorem which characterizes the advertiser's maximum  revenue. The main idea of the theorem is that the maximum budget-constrained revenue can in principle be computed by solving a complex nonlinear optimization problem (typically nonconvex). 
The theorem will later be used for our analysis. 
In the theorem, the parameter $V\geq1$ is a constant introduced for our later analysis. 
\begin{theorem}\label{theorem:profit}
The maximum time average revenue $\text{Profit}_{\text{av}}^*$ satisfies $V\text{Profit}^*_{\text{av}}\leq\Phi^*$, where $\Phi^*$ is the optimal value of the following optimization problem. 
\begin{eqnarray}
\hspace{-.4in}&&\max: \,\, \Phi\triangleq\sum_{n}\frac{V \sum_{h=0}^{\infty}\alpha_{n}^{h}G_{n}(p^{h}_n, m^{h}_n) }{  \sum_{h=0}^{\infty}\alpha_{n}^{h}F_{n}(p^{h}_n, m^{h}_n) + \sum_{h=0}^{\infty}\alpha_{n}^{h}T^{\text{fr}, h}_n} \label{eq:opt-stat}\\
\hspace{-.4in}&&\,\,\,\,\text{s.t.}\,\,\,\,\sum_{n}\frac{ \sum_{h=0}^{\infty}\alpha_{n}^{h}p^{h}_n}{ \sum_{h=0}^{\infty}\alpha_{n}^{h}F_{n}(p^{h}_n, m^{h}_n) +  \sum_{h=0}^{\infty}\alpha_{n}^{h}T^{\text{fr}, h}_n }\leq B_{\text{av}}\nonumber\\
\hspace{-.4in}&&\qquad\,\,\,\, \,\,\sum_{h=0}^{\infty}\alpha_{n}^{h}=1, \alpha_{n}^h\geq0,\forall \, n\nonumber\\
\hspace{-.4in}&&\qquad\,\,\,\, \,\,(p^{h}_n, T^{\text{fr},h}_n)\in\script{P}_n, m^{h}_n\in\script{M}_n, \,\,\forall\,n, h. \quad\Box\nonumber
\end{eqnarray}
\end{theorem}
\begin{proof}
The proof can be done by using an argument based on Caratheodory's theorem as in \cite{neelynowbook}. Omitted for brevity. 
\end{proof}
In Theorem \ref{theorem:profit}, the values $\{\alpha_n^h\}_{h=0}^{\infty}$  can be viewed as the probability of choosing the investment-freeze-configuration triple $(p^{h}_n, T^{\text{fr},h}_n, m^{h}_n)$ in every frame for search site $S_n$. With this interpretation, we immediately have the following corollary. 
\begin{coro}\label{coro:opt}
There exists an optimal stationary and randomized investment strategy $\Pi_S$, which works as follows: for each search site $S_n$, after each frame, chooses an investment-freeze-configuration triple $\{(p^{h}_n, T^{\text{fr},h}_n, m^{h}_n)\}_{h=0}^{\infty}$ with probabilities $\{\alpha_{n}^{h}\}_{h=0}^{\infty}$, independent of  other sites. In particular, $\Pi_S$ achieves the followings: 
\begin{eqnarray*}
\sum_n\frac{V\expect{R_n[k]}}{\expect{T_n^{\text{ad}}[k] + T_n^{\text{fr}}[k]}}  &=& \Phi^*, \\
\sum_n\frac{\expect{p_n[k]}}{\expect{T_n^{\text{ad}}[k] + T_n^{\text{fr}}[k]  }} &\leq& B_{\text{av}}. \quad \Box
\end{eqnarray*}
\end{coro}

However, as discussed above, due to the non-convex nature of problem (\ref{eq:opt-stat}), directly solving for such a randomized policy is very difficult. In the following, we instead develop an online algorithm, which allows the advertiser to \emph{asynchronously} update his investment profile at different search sites and achieve near-optimal profit. We will see that our algorithm only requires \emph{greedy} updates of each site and can be implemented efficiently in practice. 

\section{Online Algorithm for Achieving Near Optimal Revenue}\label{section:ai}
In this section, we develop an investment strategy to achieve a near optimal performance for the advertiser. Our algorithm is based on the renewal and asynchronous system optimization technique developed in \cite{neelymcbook} and \cite{neely-asyn}. 

\subsection{The Virtual Deficit Queue}
To start, we first define the notion of  \emph{decision points}. A time $t$ is called a decision point if it marks the beginning of an advertising interval for any site. Then, we use $\{t_d\}_{d=0}^{\infty}$ to denote the set of decision points. We see then the timeline is divided into disjoint intervals by the decision points. Fig. \ref{fig:service-idle-frame} shows five such points $\{t_0, t_1, ..., t_4\}$. Since each frame will be at least $T^{\min}$ long, we see that $\lim_{d\rightarrow\infty}t_d=\infty$ w.p $1$. 

We use $k_n(t)$ to denote the number of advertising intervals that site $n$ has completed up to time $t$. Note that $k_n(t)$ also denotes  the index of the current advertising interval the search site is in at time $t$. 
We denote $\delta_d\triangleq t_{d+1}-t_d$ the length of the interval between decision points $t_d$ and $t_{d+1}$. 
For our analysis, it is also useful to define the notions of \emph{effective budget consumption rate} $a_n[t_d]$, and \emph{effective revenue generating rate} $r_n[t_d]$ over the time interval $[t_d, t_{d+1})$ for each site $S_n$ as follows: 
\begin{eqnarray}
\hspace{-.5in}&&a_n[t_d] \triangleq \frac{p_n[k_n(t_d)]}{F_{n}[k_n(t_d)]+T^{\text{fr}}_{n}[k_n(t_d)]}, \label{eq:budget-consumption}\\
\hspace{-.5in}&&r_n[t_d] \triangleq \frac{G_{n}[k_n(t_d)]}{F_{n}[k_n(t_d)]+T^{\text{fr}}_{n}[k_n(t_d)]}. \label{eq:revenue-gen}
\end{eqnarray} 
Here $a_n[t_d]$ and $r_n[t_d]$ roughly denote the average budget consumption rate and the average revenue generation rate. Note that we use the \emph{expected} value of the advertising interval length in both (\ref{eq:budget-consumption}) and (\ref{eq:revenue-gen}). 

Now we define the \emph{deficit queue process} $Q(t_d)$, to keep track of the budget deficit due to temporal violation of the average budget constraint. $Q(t_d)$ evolves according to the following dynamics: 
\begin{eqnarray}
Q(t_{d+1}) = \max\big[ Q(t_d) -  \mu[t_d], 0  \big]  + A[t_d],\label{eq:deficit-queue}
\end{eqnarray}
with $Q(t_0)=0$.  
Here $A[t_d]$ denotes the \emph{aggregate} budget consumed by all the sites during $[t_d, t_{d+1})$, i.e., 
\begin{eqnarray}
A[t_d] =\delta_d \sum_{n}a_n[t_d]. \label{eq:A-def}
\end{eqnarray}
Similarly, $\mu[t_d]$ denotes the amount of ``allowed'' average budget consumption, i.e., 
 \begin{eqnarray}
\mu[t_d]=B_{\text{av}}(t_{d+1}-t_d)=\delta_dB_{\text{av}}. \label{eq:mu-def}
\end{eqnarray}  
We notice that since the update of $Q(t_d)$ uses expected lengths of the actual advertising intervals, $Q(t_d)$ is an \emph{approximation} of the true deficit. However, we will show later that guaranteeing the stability of $Q(t_d)$ also ensures the actual  budget constraint (\ref{eq:budget-cond}). 

For notational convenience, let us also define $c_{\max}=T^{\max}\max[\sum_n\frac{p_n^{\max}}{T_n^{\min}},  B_{\text{av}}]$. We see then $A[t_d], \mu[t_d]\leq c_{\max}$ for all $t_d$.

\vspace{-.1in}
\subsection{The Dynamic \textsf{Ad-Investment (AI)} Algorithm}
Now we present our investment algorithm. In the algorithm, $\script{S}(t_d)$ denotes the set of search engine sites that just finish a frame and will start a new advertising interval at time $t_d$. 
The algorithm works as follows: 

\underline{\textsf{Ad-Investment (AI):}} At each time $t_d$, denote $k_n=k_n(t_d)$  the frame index of $S_n$ at time $t_d$. For each search site $S_n\in\script{S}(t_d)$, observe $Q(t_d)$ and  perform the following:  
\begin{enumerate}
\item Choose $p_n[k_n]$, $m_n[k_n]$ and $T^{\text{fr}}_{n}[k_n]$ that solve the following:
\begin{eqnarray}
\hspace{-.8in}&&\max:\,\,\, \Psi_n(t_d) \triangleq\frac{VG_{n}(p_n[k_n], m_n[k_n]) - Q(t_d)p_n[k_n]}{F_{n}(p_n[k_n], m_n[k_n])+T^{\text{fr}}_{n}[k_n] }  \label{eq:ai-obj}\\
\hspace{-.8in}&&\,\,\text{s.t.} \qquad (p_n[k_n], T^{\text{fr}}_{n}[k_n])\in\script{P}_n, \, m_n[k_n]\in\script{M}_n. \nonumber
\end{eqnarray}
Then, update the investment and the advertising configuration at $S_n$ with the chosen actions. 
\item Update $Q(t_d)$ according to (\ref{eq:deficit-queue}).  $\Diamond$
\end{enumerate}
Denote $\Psi^*_n(t_d)$ the optimal value of (\ref{eq:ai-obj}). We note that since $\script{P}_n$ includes the option $p_n=0$ and $G_n(0, \cdot)=0$, at the optimal solution we always have $\Psi^*_n(t_d)\geq0$. 
We also note that the algorithm involves solving problem (\ref{eq:ai-obj}) for each search site \emph{separately}. Hence, the advertiser can easily update his configuration at any site when the current frame completes.  This is quite different from many existing algorithms that update all control actions simultaneously. 
Also note that (\ref{eq:ai-obj}) can often be solved efficiently in practice, especially when the actions sets $\script{P}_n$  and $\script{M}_n$ are finite. 
Finally, note that in every step, we use $F_n(p, m)$ and $G_n(p, m)$ for decision making. This requires the advertiser to have accurate prior estimation of the system dynamics. As we will see in Section \ref{section:approx} that \textsf{AI} can also be implemented with imperfect knowledge of the system dynamics,  and still achieve good performance. 

\subsection{Performance Analysis}
Define the Lyapunov function $L(t_d) = \frac{1}{2}Q(t_d)^2$. Then, we define a $1$-slot Lyapunov drift as follows: 
\begin{eqnarray}
\Delta(t_d) \triangleq  L(t_{d+1}) -  L(t_d) .   \label{eq:drift-def}
\end{eqnarray}
Using the queueing dynamic equation (\ref{eq:deficit-queue}), we have the following lemma. 
\begin{lemma}\label{lemma:drift}
At any time $t_d$, we have:  
\begin{eqnarray}
\Delta(t_d) \leq C_0 - Q(t_d)\big(\delta_dB_{\text{av}} - A[t_d] \big), \label{eq:drift-ineq}
\end{eqnarray} 
where $C_0\triangleq \frac{1}{2}(T^{\max})^2\big[(\sum_n\frac{p_n^{\max}}{T_n^{\min}})^2+B^2_{\text{av}}\big]$. $\Box$
\end{lemma}
\begin{proof}
See Appendix A. 
\end{proof}
Adding the term $-V\delta_d\sum_nr_n[t_d]$ to both sides of (\ref{eq:drift-ineq}) and rearranging terms, we obtain: 
\begin{eqnarray}
\hspace{-.4in}&&\Delta(t_d) -V \delta_d\sum_nr_n[t_d]\nonumber\\
\hspace{-.4in}&&\qquad\leq C_0 - Q(t_d)\delta_dB_{\text{av}} \nonumber \\
\hspace{-.4in}&&\qquad\qquad\,\,\,- \delta_d\sum_n \frac{VG_{n}[k_n(t_d)] - Q(t_d)p_n[k_n(t_d)]}{ F_{n}[k_n(t_d)]+T^{\text{fr}}_{n}[k_n(t_d)]} \nonumber\\
\hspace{-.4in}&&\qquad\leq C_0 - Q(t_d)\delta_dB_{\text{av}} - \delta_d\sum_n\Psi_n(t_d). \label{eq:drift-ineq-2}
\end{eqnarray} 
From (\ref{eq:drift-ineq-2}), one sees that the \textsf{AI} algorithm is indeed constructed to maximize $\Psi_n(t_d)$ on the right-hand-side (RHS) of the drift inequality for every $S_n\in\script{S}(t_d)$. 
The following lemma shows that \textsf{AI} also approximately maximizes $\sum_n\Psi_n(t_d)$. 
\begin{lemma}\label{lemma:ai-approx}
Under the \textsf{AI} algorithm, at every time $t_d$, we have:  
\begin{eqnarray}
\sum_n\Psi^\textsf{AI}_n(t_d)\geq \max_{\Pi} \bigg(\sum_n\Psi^{\Pi}_n(t_d)\bigg) - C_1. \label{eq:ai-approx}
\end{eqnarray}
Here $\Pi$ is any feasible investment policy in choosing $p_n[k_n(t_d)]$, $m_n[k_n(t_d)]$, and $T_n^{\text{fr}}[k_n(t_d)]$ at time $t_d$, and $C_1=2T^{\max}\sum_n\frac{p_n^{\max}}{T_n^{\min}}$ is a $\Theta(1)$ constant independent of $V$. 
\end{lemma}
\begin{proof}
See Appendix B. 
\end{proof}
Lemma \ref{lemma:ai-approx} is important for our analysis. In practice, due to the dynamics in the advertising process, sites typically do not complete their frames at the same time. Hence, at every instance of time, only some search sites need to update their configurations. In this case, if instead a synchronous scheme is applied, some sites will have to remain idle while advertising is still going on at others. This will inevitably result in resources being wasted. 
Lemma \ref{lemma:ai-approx} shows that \textsf{AI} roughly minimizes the RHS of the drift inequality (\ref{eq:drift-ineq-2}) even with asynchronous updates. This guarantees the near-optimal performance of \textsf{AI}, as  summarized in the following theorem. 
\begin{theorem}\label{theorem:ai-per} 
The \textsf{AI} algorithm achieves the following: 
\begin{enumerate}
\item[a)] The average revenue under \textsf{AI}, denoted by $\text{Profit}^{\textsf{AI}}_{\text{av}}$, satisfies: 
\begin{eqnarray}
\text{Profit}^{\textsf{AI}}_{\text{av}} \geq \text{Profit}^{*}_{\text{av}} -\frac{C_1}{V} - \frac{C_0}{VT^{\min}}. 
\end{eqnarray}
Here $C_0$ and $C_1$ are $\Theta(1)$ constants defined in Lemmas \ref{lemma:drift} and \ref{lemma:ai-approx}. 
\item[b)] At every time $t_d$, we have: 
\begin{eqnarray}
Q(t_d)\leq V\nu + 2c_{\max}. \label{eq:q-bound}
\end{eqnarray}
Moreover, the average budget constraint (\ref{eq:budget-cond}) is satisfied with probability $1$. 
\end{enumerate} 
\end{theorem}
\begin{proof}
See Appendix C. 
\end{proof} 
By tuning the value of $V$, Theorem \ref{theorem:ai-per} implies that \textsf{AI} can indeed achieve a revenue that is within $O(1/V)$ of the maximum, for any $V\geq1$, with a tradeoff that $Q(t_d)=O(V)$. Also notice that (\ref{eq:q-bound}) does not automatically ensure that the actual deficit is deterministically bounded. This is because the actual frame length is a random variable while $Q(t_d)$ is updated with the expected values. Despite this approximation, Part b) of Theorem \ref{theorem:ai-per}  shows that so long as $Q(t_d)$ is stable,  (\ref{eq:budget-cond}) will be satisfied with probability $1$. 


\section{Performance under Imperfect Estimation}\label{section:approx}
So far, we have assumed that the advertiser knows exactly the expected revenue that will be generated during a frame and the expected duration of the frame. 
In this section, we consider the more general case where the advertiser only has imperfect estimations of the functions $G_n$ and $F_n$, and quantify the impact of such inaccuracy on  the performance of the \textsf{AI} algorithm. 

Specifically, we assume that now the advertiser does not have the perfect knowledge of $G_n(p, m)$ and $F_n(p, m)$. Instead, he uses his own estimated  functions 
$\hat{G}_n(p, m)$ and $\hat{F}_n(p, m)$ for algorithm design and decision making. That is, $Q(t_d)$ and $\Psi_n(t_d)$ will be defined with $\hat{G}_n(p, m)$ and $\hat{F}_n(p, m)$, and (\ref{eq:ai-obj}) will be solved with them as well. 
In order to quantify the impact of the estimation errors, we first introduce the quality index of the estimations, $\rho_G$ and $\rho_F$. That is, we say that the advertiser has an estimation quality $(\rho_G, \rho_F)$ if for all $S_n$, 
\begin{eqnarray}
|\hat{G}_n(p, m) - G_n(p, m)|&\leq& \rho_G G_n(p, m), \,\,\forall\, p, m,\label{eq:g-est}\\
|\hat{F}_n(p, m) - F_n(p, m)|&\leq& \rho_F F_n(p, m), \,\,\forall\, p, m. \label{eq:f-est}
\end{eqnarray}
Thus, $\rho_F$ and $\rho_G$ capture the inaccuracy levels of the estimations. 
We now have the following theorem regarding the performance of the \textsf{AI} algorithm in this case. 
\begin{theorem}\label{theorem:error}
Suppose the advertiser has an estimation quality $(\rho_G, \rho_F)$ where $\rho_G, \rho_F\in[0, 1)$. Then, the \textsf{AI} algorithm achieves the following: 
\begin{enumerate}
\item[a)] Revenue performance: 
\begin{eqnarray}
\hspace{-.3in}\text{Profit}^{\textsf{AI}}_{\text{av}} \geq\frac{1-\rho_F}{(1+\rho_F)(1+\rho_G)}\text{Profit}^{*}_{\text{av}}- \Theta(\frac{1}{V}) - \frac{N\rho_GG^{\max}}{(1+\rho_G)T^{\min}}. \nonumber
\end{eqnarray}
\item[b)] Budget control: At every time $t_d$, we have: 
\begin{eqnarray}
\hspace{-.3in}Q(t_d)\leq V(1+\rho_G)\nu + 2\hat{c}_{\max}. \label{eq:q-bound2}
\end{eqnarray}
Here $\hat{c}_{\max}$ is defined as the maximum change of $Q(t_d)$ defined with $\hat{G}_n(p, m)$ and $\hat{F}_n(p, m)$, i.e.,  
\begin{eqnarray*}
\hspace{-.3in}\hat{c}_{\max}=(1+\rho_F)T^{\max}\max[\sum_n\frac{p_n^{\max}}{(1-\rho_F)T_n^{\min}}, B_{\text{av}}]. 
\end{eqnarray*}
Moreover, the average budget constraint (\ref{eq:budget-cond}) is satisfied with probability $1$ if we replace $B_{\text{av}}$ with $(1+\rho_F)B_{\text{av}}$.  
\end{enumerate} 
\end{theorem}
\begin{proof}
See Appendix D. 
\end{proof}
We note Part b) of Theorem \ref{theorem:error} indicates that  the   budget constraint may be slightly violated if the estimations of $G_n$ and $F_n$ are inaccurate. While such violation is typical in this case, Theorem \ref{theorem:error} shows that as long as one uses a  scaled budget in the algorithm, i.e., use $B_{\text{av}}/(1+\rho_F)$, the original advertisement budget can easily be guaranteed. 
This thus provides useful guidelines for the actual algorithm implementation.  
Also note that Theorem \ref{theorem:error} applies to more general cases where the system's dynamics and  the advertiser's estimation quality both vary over time. 
In this case, if the advertiser's estimation quality stays better than a certain quality $(\rho_G, \rho_F)$ after some finite learning time, then one can invoke Theorem \ref{theorem:error} and conclude the performance of the algorithm. 

\section{Simulation}\label{section:sim}
In this section, we provide simulation results for our algorithm. We consider the case when the system has two sites $\{S_1, S_2\}$. For each site $S_n$, we assume for simplicity that $m_n$ only denotes the maximum pay-per-click. 

The $G_n(p, m)$ and $F_n(p, m)$ functions are given by: 
\begin{eqnarray}
F_n(p, m) = \frac{\kappa_np}{m}, \quad G_n(p, m) = \gamma_n\sqrt{\frac{p}{m}}m^{q}. 
\end{eqnarray}
The function $F_n$ is chosen to model the fact that the expected duration of the advertising interval is proportional to $\frac{p}{m}$, i.e., roughly the number of times the advertiser wins the ad opportunity. 
The choice of function $G_n$ assumes that whenever an ad opportunity arises, the advertiser wins the opportunity with $m^q$ success probability, and that it roughly takes $p/m$ ad displays to generate $\sqrt{p/m}$ revenue. 
For both $S_n$, we assume that $m_n=\{0.1, 0.2\}$ and that the feasible (investment-freeze) set is given by: 
\begin{eqnarray*}
\hspace{-.2in}&&\script{P}_n=\{(p=0, T=5), (p=5, T=0), (p=5, T=5),\\
\hspace{-.2in}&&\qquad\qquad\qquad\qquad \qquad (p=10, T=0), (p=10, T=5)\}. 
\end{eqnarray*}
For each action, we assume that the actual revenue is uniformly distributed in $[0.8G_n, 1.2G_n]$. Also, the advertising interval length is uniformly distributed in $[0.8F_n, 1.2F_n]$. 
We set $\kappa_1=1$ and $\kappa_2=2$; $\gamma_1=1$ and $\gamma_2=2$;  $B_{\text{av}}=0.2$ and we use $q=0.2$. We simulation the algorithm for $V\in\{5, 10, 20, 50, 100, 200\}$. 
\begin{figure}[cht]
\centering
\vspace{-.08in}
\includegraphics[height=1.8in, width=3.4in]{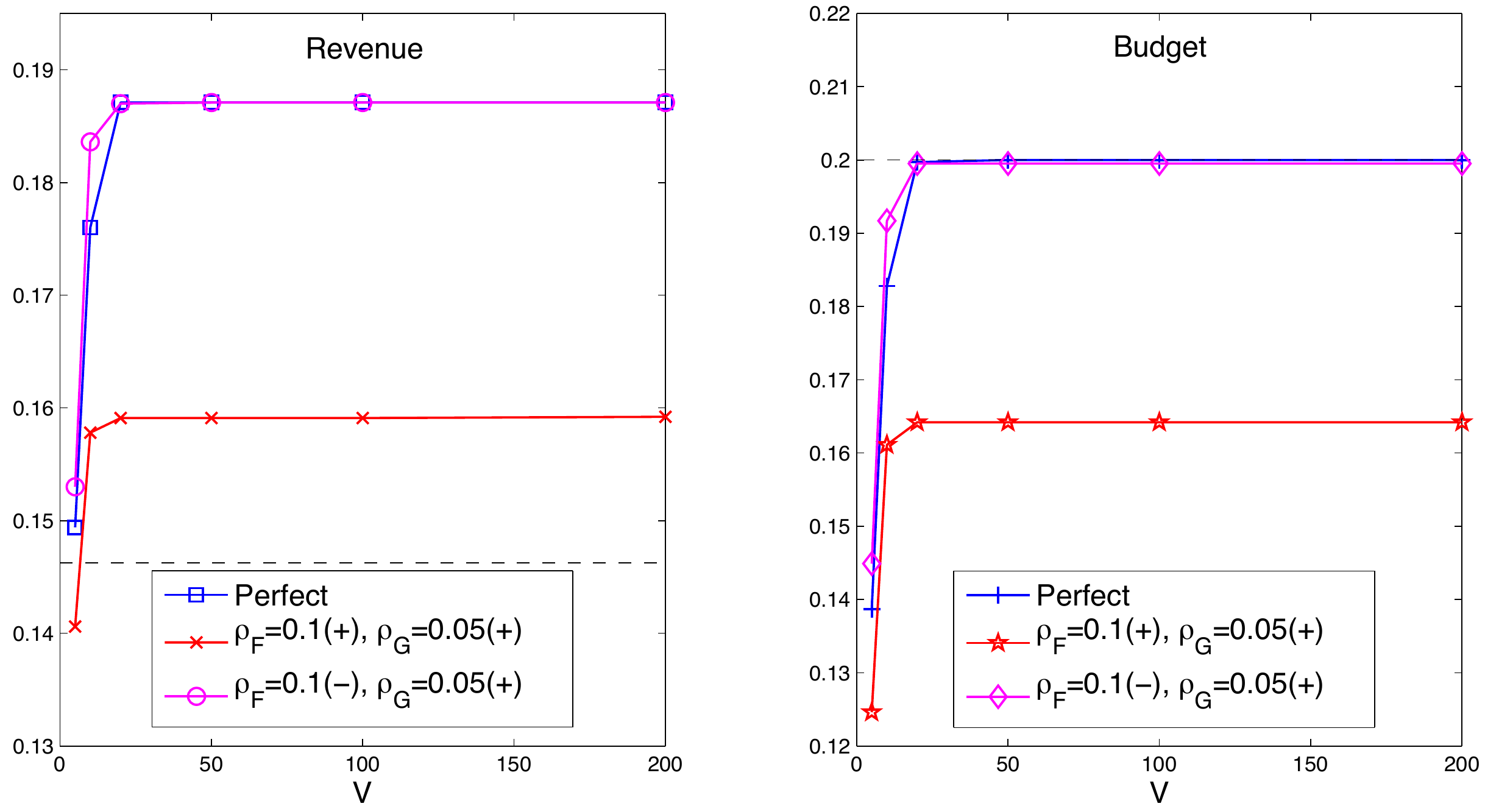}
\vspace{-.05in}
\caption{Performance of the \textsf{AI} algorithm. The legend $\rho_F=0.1(\pm), \rho_G=0.05(\pm)$ means that $\hat{F_n}(p, m)=(1\pm0.1)F_n(p, m)$ and $\hat{G_n}(p, m)=(1\pm0.05)G_n(p, m)$. 
We see that when there is estimation error, using a scaled budget value easily guarantees the original advertising budget constraint. }\label{fig:performance-plot}
\vspace{-.05in}
\end{figure}

First, we see from Fig. \ref{fig:performance-plot} that as  the $V$ value increases, the average revenue quickly converges to the optimal when there is no estimation errors. In the imperfect estimation case, we simulate \textsf{AI} with a scaled budget $B_{\text{av}}/(1+\rho_F)$. In this case, we see from the right plot that the actual budget never exceeds $B_{\text{av}}=0.2$. Also, the left plot shows that the revenue performance of \textsf{AI} is always above $\frac{1-\rho_F}{(1+\rho_F)(1+\rho_G)}\text{Profit}^*_{\text{av}}$ (indicated by the black dotted line), which implies that it is also always no less than $\frac{1-\rho_F}{(1+\rho_F)(1+\rho_G)}$ fraction of the optimal revenue with budget constraint $B_{\text{av}}/(1+\rho_F)$. These results are consistent with Theorems \ref{theorem:ai-per} and \ref{theorem:error}. 

\begin{figure}[cht]
\centering
\includegraphics[height=1.8in, width=3.4in]{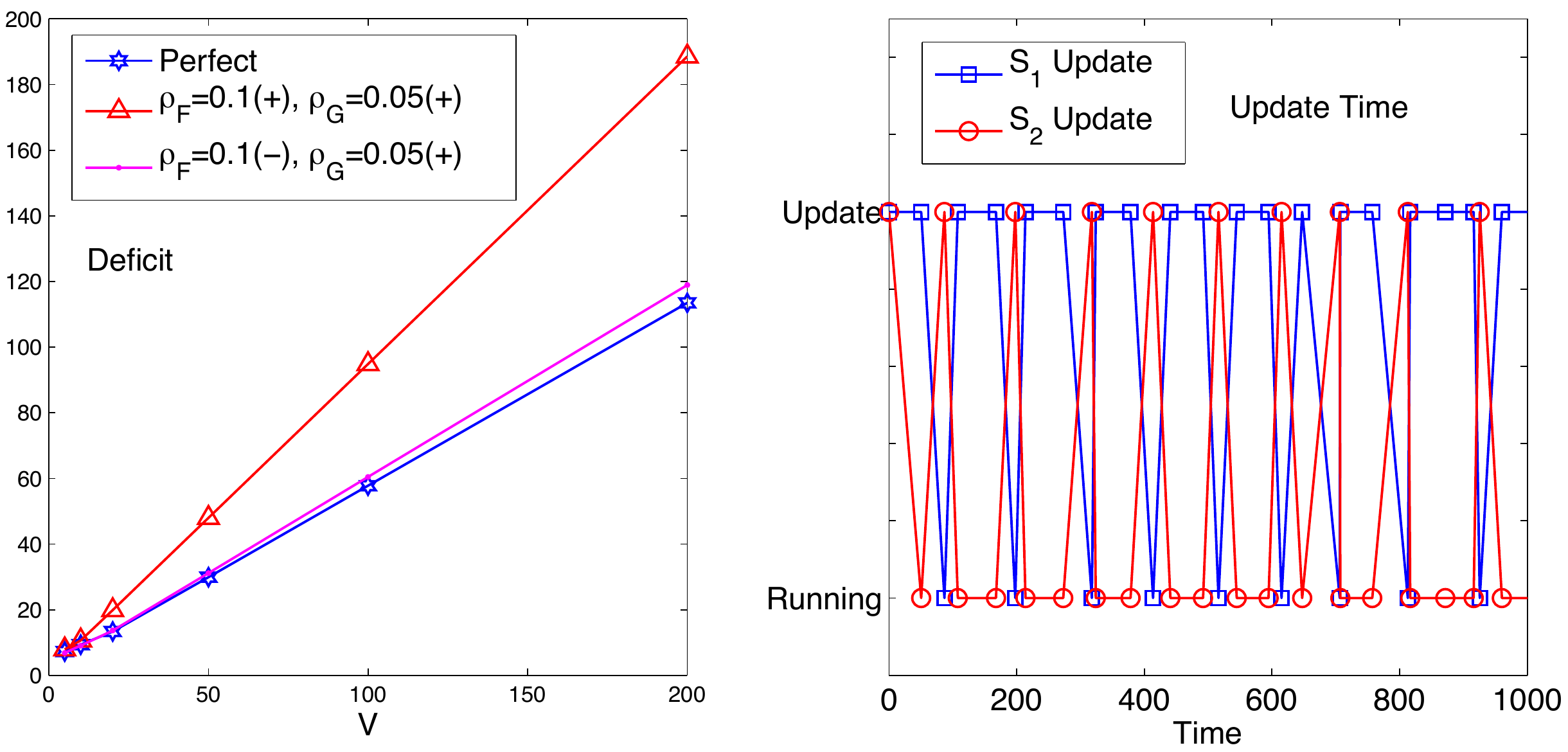}
\vspace{-.05in}
\caption{The left plot shows that the average deficit under the \textsf{AI} algorithm increases linearly in the $V$ parameter. The right plot shows that the advertiser updates the advertisement and configuration of the sites asynchronously. }\label{fig:performance-plot2}
\vspace{-.15in}
\end{figure}

To evaluate the effect of $V$ on budget violation, Fig. \ref{fig:performance-plot2} also plots the average deficit queue size. We see that it increases linearly in $V$. The right plot of Fig. \ref{fig:performance-plot2}  also shows the time when the two sites update their configurations, for the case when $V=20$ and there is no estimation error. It can be seen that under \textsf{AI}, the site updates take place \emph{asynchronously}. 

\section{Conclusion}\label{section:conclusion}
In this paper, we design optimal investment strategies for a single advertiser who utilizes multiple sponsored search engine sites simultaneously for advertising. An \emph{asynchronous} investment strategy \textsf{Ad-Investment} (\textsf{AI}) is proposed. The algorithm allows the advertiser to \emph{asynchronously} update his investment configuration at each site. We show that the algorithm achieves an average revenue that is within $O(\epsilon)$ of the optimal, for any $\epsilon>0$, while guaranteeing that the temporal budget violation is $O(1/\epsilon)$. We also analyze the performance of \textsf{AI} under imperfect system estimation and show that it is robust against estimation errors.

\vspace{-.06in}
\section*{Appendix A -- Proof of Lemma \ref{lemma:drift}}
Here we prove Lemma \ref{lemma:drift}. 
\begin{proof}
Squaring both sides of (\ref{eq:deficit-queue}) and using the fact that $(\max[x, 0])^2\leq x^2$ for all $x\in\mathbb{R}$, we have: 
\begin{eqnarray}
\hspace{-.3in}&&Q^2(t_{d+1})\leq Q^2(t_d) + (\mu[t_d])^2 + (A[t_d])^2\nonumber \\
\hspace{-.3in}&&\qquad\qquad\qquad\qquad\,\,- 2 Q(t_d) \big( \mu[t_d] - A[t_d]  \big). \label{eq:drift-0}
\end{eqnarray}
Now using the facts that $a_n[t_d]\leq p_n^{\max}/T_n^{\min}$ and $\delta_d\leq T^{\max}$, we see that at any time $t_d$,  
\begin{eqnarray}
\hspace{-.3in}&&\frac{1}{2}(A[t_d])^2 + \frac{1}{2}(\mu[t_d])^2 \nonumber\\
\hspace{-.3in}&&\qquad\quad\leq C_0 \triangleq \frac{1}{2}(T^{\max})^2\big[(\sum_n\frac{p_n^{\max}}{T_n^{\min}})^2+B^2_{\text{av}}\big]. \label{eq:c0}
\end{eqnarray}
Multiplying both sides of (\ref{eq:drift-0}) by $\frac{1}{2}$ and using the definition of $C_0$, we sees that the lemma follows. 
\end{proof}

\vspace{-.06in}
\section*{Appendix B -- Proof of Lemma \ref{lemma:ai-approx}}
\begin{proof}
To prove the lemma, we first recall that the $\Psi_n(t_d)$ function is  defined as:
\begin{eqnarray}
\Psi_n(t_d) \triangleq\frac{VG_{n}(p_n[k_n], m_n[k_n]) - Q(t_d)p_n[k_n]}{F_{n}(p_n[k_n], m_n[k_n])+T^{\text{fr}}_{n}[k_n] }
\end{eqnarray}
Here $k_n=k_n[t_d]$. 

If at time $t_d$, $S_n\in\script{S}(t_d)$, then $\Psi_n(t_d)$ is maximized over all possible pure strategies, i.e., strategies that choose a given $(p_n[k], m_n[k], T^{\text{fr}}_{n}[k])$ with probability $1$. Now notice that for any $c_1, d_1, c_2, d_2$ where $d_1, d_2\neq0$ and $\theta\in[0, 1]$, we have: 
\begin{eqnarray}
\frac{\theta c_1+(1-\theta)c_2}{\theta d_1+(1-\theta)d_2} = \eta \frac{c_1}{d_1} + (1-\eta) \frac{c_2}{d_2}, 
\end{eqnarray}
for some $\eta\in[0, 1]$ \cite{boydconvexopt}. Therefore, if a policy $\Pi$ mixes two pure strategies $\Pi_1$ and $\Pi_2$, i.e., uses two triples $(p_1, m_1, T_1)$ and $(p_2, m_2, T_2)$ with probabilities $a_1>0$ and $a_2>0$, where $a_1+a_2=1$, then there exists $\eta_\Pi\in[0, 1]$ that: 
\begin{eqnarray}
\Psi_n^{\Pi}(t_d) = \eta_{\Pi}\Psi_n^1(t_d) + (1-\eta_{\Pi})\Psi_n^2(t_d)\leq \Psi^{\textsf{AI}}_n(t_d). \label{eq:mixed}
\end{eqnarray} 
By induction, we see that $\Psi_n(t_d)$ is maximized under \textsf{AI} over all strategies for all $S_n\in\script{S}(t_d)$. 

It remains to show that for any other $S_n\notin\script{S}(t_d)$, the function $\Psi_n(t_d)$ is also \emph{approximately} maximized. 
Consider an $S_n\notin\script{S}(t_d)$. Let $\tilde{t}_n$ be the last time the advertiser updates his configuration at $S_n$ before $t_d$, i.e., $S_n\in\script{S}(\tilde{t}_n)$ but $S_n\notin \script{S}(t'_{d})$ for all $\tilde{t}_n< t'_d\leq t_d$. Note that $t_d-\tilde{t}_n\leq T^{\max}$. 
This is because the maximum frame length of any site $S_n$ is bounded by $T^{\max}$. 
Note that at time $t_d$, the site $S_n$ is still operating under the actions chosen at time $\tilde{t}_n$, which is based on the queue size $Q(\tilde{t}_n)$. Denote the actions chosen then by $\tilde{p}_n$, $\tilde{m}_n$ and $\tilde{T}^{\text{fr}}_n$. 

Denote $p^*_n$, $m^*_n$ and $T^{\text{fr}*}_n$ the actions chosen at time $t_d$ by solving  (\ref{eq:ai-obj})  with the backlog $Q(t_d)$. Note that these actions are not actually implemented because the advertiser will not update the configuration of $S_n$ at time $t_d$. Then, we have: 
\begin{eqnarray}
&&\frac{VG_{n}(\tilde{p}_n, \tilde{m}_n) - Q(t_d)\tilde{p}_n}{F_{n}(\tilde{p}_n, \tilde{m}_n)+\tilde{T}^{\text{fr}}_n}\nonumber\\
& =  &\frac{VG_{n}(\tilde{p}_n, \tilde{m}_n) - Q(\tilde{t}_n)\tilde{p}_n}{F_{n}(\tilde{p}_n, \tilde{m}_n)+\tilde{T}^{\text{fr}}_n} - \frac{ (Q(t_d)-Q(\tilde{t}_n))\tilde{p}_n}{F_{n}(\tilde{p}_n, \tilde{m}_n)+\tilde{T}^{\text{fr}}_n}\nonumber\\
& \stackrel{(a)}{\geq} & \frac{VG_{n}(p^*_n, m^*_n) - Q(\tilde{t}_n)p^*_n}{F_{n}(p^*_n, m^*_n)+T^{\text{fr}*}_n} + \frac{ (Q(t_d)-Q(\tilde{t}_n))\tilde{p}_n}{F_{n}(\tilde{p}_n, \tilde{m}_n)+\tilde{T}^{\text{fr}}_n}\nonumber\\
&=& \frac{VG_{n}(p^*_n, m^*_n) - Q(t_d)p^*_n}{F_{n}(p^*_n, m^*_n)+T^{\text{fr}*}_n} - \frac{ (Q(t_d)-Q(\tilde{t}_n))\tilde{p}_n}{F_{n}(\tilde{p}_n, \tilde{m}_n)+\tilde{T}^{\text{fr}}_n}\nonumber\\
&& + \frac{  (Q(t_d)-Q(\tilde{t}_n))p^*_n}{F_{n}(p^*_n, m^*_n)+T^{\text{fr}*}_n}.  \label{eq:diff-0}
\end{eqnarray}
Here the inequality (a) is because $\tilde{p}_n$, $\tilde{m}_n$ and $\tilde{T}^{\text{fr}}_n$ maximize (\ref{eq:ai-obj}) given $Q(\tilde{t}_n)$. 

Using the queueing dynamic equation (\ref{eq:deficit-queue}), we see that: 
\begin{eqnarray*}
Q(t_d)-(t_d-\tilde{t}_n)c_{\max}\leq Q(\tilde{t}_n) \leq Q(t_d) +(t_d-\tilde{t}_n)c_{\max}. 
\end{eqnarray*}
Together with the fact that $t_d-\tilde{t}_n\leq T^{\max}$, we have: 
\begin{eqnarray*}
|Q(t_d) - Q(\tilde{t}_n)| \leq  T^{\max}c_{\max}. 
\end{eqnarray*}
Using this in (\ref{eq:diff-0}), we get: 
\begin{eqnarray}
\frac{VG_{n}(\tilde{p}_n, \tilde{m}_n) - Q(t_d)\tilde{p}_n}{F_{n}(\tilde{p}_n, \tilde{m}_n)+\tilde{T}^{\text{fr}}_n}  \geq  \Psi^*_n(t_d) - C_{1n}, \label{eq:ineq-pure}
\end{eqnarray}
where $ \Psi^*_n(t_d)$ is the value of $\Psi_n(t_d)$ with the actions $p^*_n$, $m^*_n$, $T^{\text{fr}*}_n$, and $Q(t_k)$, and $C_{1n} \triangleq \frac{2T^{\max}c_{\max} p_n^{\max}}{T_n^{\min}}$. This shows that even for $S_n\notin\script{S}(t_d)$, $\Psi_n(t_d)$ is approximately maximized. Now by defining $C_1=\sum_nC_{1n}$ and using a similar argument as in (\ref{eq:mixed}), we see that the lemma follows. 
\end{proof}

\vspace{-.1in}
\section*{Appendix C -- Proof of Theorem \ref{theorem:ai-per}}
Here we prove Theorem \ref{theorem:ai-per}. 
\begin{proof} (Theorem \ref{theorem:ai-per}) 
(Revenue Performance) We first prove the revenue performance. To do so, we recall the drift inequality (\ref{eq:drift-ineq-2}) as follows: 
\begin{eqnarray}
\hspace{-.4in}&&\Delta(t_d) -V \delta_d\sum_nr_n[t_d]\nonumber\\
\hspace{-.4in}&&\qquad\leq C_0 - Q(t_d)\delta_dB_{\text{av}} - \delta_d\sum_n\Psi_n(t_d). \label{eq:drift-ineq-3}
\end{eqnarray} 
By Lemma \ref{lemma:ai-approx}, we note that \textsf{AI} maximizes the term $\sum_n\Psi_n(t_d)$ over all possible investment policies for up to a finite constant $C_1$. Thus, we can plug into the RHS of (\ref{eq:drift-ineq-3}) any control algorithm and the following holds: 
\begin{eqnarray}
\hspace{-.4in}&&\Delta(t_d) -V\delta_d \sum_nr_n[t_d]\nonumber\\
\hspace{-.4in}&&\qquad\leq C_0 +C_1\delta_d - Q(t_d)\delta_dB_{\text{av}} - \delta_d\sum_n\Psi^{\text{ALT}}_n(t_d). \label{eq:drift-ineq-4}
\end{eqnarray} 
Therefore, we plug the randomized and stationary policy $\Pi_S$ in Corollary \ref{coro:opt} into (\ref{eq:drift-ineq-3}) and obtain: 
\begin{eqnarray}
\hspace{-.4in}&&\Delta(t_d) -V \delta_d\sum_nr_n[t_d] \leq C_0  +\delta_dC_1- \delta_d\Phi^*. \label{eq:drift-ineq-4}
\end{eqnarray} 
Taking expectations on both sides, carrying out a telescoping sum from $d=0$ to $D-1$, and rearranging terms, we obtain: 
\begin{eqnarray*}
\hspace{-.4in}&&V\sum_{d=0}^{D-1}\expect{\delta_d\sum_nr_n[t_d]} \nonumber\\
\hspace{-.4in}&&\qquad\qquad\geq \sum_{d=0}^{D-1}\expect{\delta_d}(\Phi^*-C_1) - DC_0 - \expect{L(t_0)}. 
\end{eqnarray*}
Dividing both sides by $V\sum_d\expect{\delta_d}$ and letting $D\rightarrow\infty$, 
\begin{eqnarray}
\hspace{-.4in}&&\lim_{D\rightarrow\infty}\frac{\sum_{d=0}^{D-1}\expect{\delta_d\sum_nr_n[t_d]}}{\sum_{d=0}^{D-1}\expect{\delta_d}}\nonumber \\
\hspace{-.4in}&&\qquad\qquad\geq  \Phi^*/V-\frac{C_1}{V} - \lim_{D\rightarrow\infty}\frac{C_0D}{V\sum_{d=0}^{D-1}\expect{\delta_d}}\nonumber\\
\hspace{-.4in}&&\qquad\qquad \geq\text{Profit}^*_{\text{av}} -\frac{C_1}{V} - \frac{C_0}{VT^{\min}}. \label{eq:ai-per}
\end{eqnarray}

It remains to show that the actual average revenue also satisfies (\ref{eq:ai-per}). 
%
%
Denote $z_n^{p, m, T}$  the long term average fraction of time that the triple $(p, m, T)$ is used at $S_n$ under \textsf{AI}. Because for any $(p, m)$, $\expect{R_n[k]}$ and $\expect{T^{\text{ad}}_n[k]}$ are bounded, we use the strong law of large numbers (SLLN) \cite{durrett_prob} to conclude:   
\begin{eqnarray} 
\hspace{-.4in}&& \sum_{n}\frac{\lim_{K\rightarrow\infty}\frac{1}{K}\sum_{k=0}^{K-1}R^{\textsf{AI}}_n[k]  }{ \lim_{K\rightarrow\infty}\frac{1}{K}\sum_{k=0}^{K-1}(T^{\text{ad}, \textsf{AI}}_n[k] + T^{\text{fr},\textsf{AI}}_n[k]  ) }\label{eq:exp2actual-0}\\
\hspace{-.4in}&& \qquad\qquad\qquad\quad=\sum_n\frac{\sum_{p, m, T}z_n^{p, m, T}G_n(p, m)}{\sum_{p, m, T}z_n^{p, m, T}[F_n(p, m) + T]}. \nonumber
\end{eqnarray}
Now fix a search site $S_n$, we see then for the $t_d$ values that are in the middle of a frame of $S_n$, $r_n[t_d]$ remains unchanged. Let $\tilde{t}_n$ be the last time $S_n$ starts a new frame before $t_D$, and let $\script{J}_n^D(p,m,T)$ be the set of frames during which the triple $(p, m, T)$ is adopted at site $S_n$ up to $t_D$. Then, 
\begin{eqnarray*} 
\hspace{-.1in}&&  \sum_{d=0}^{D-1}\expect{\delta_dr_n[t_d]}  \\
\hspace{-.1in}&=&  \sum_{p, m, T}\frac{G_n(p, m)}{F_n(p, m)+T}\expect{\sum_{j\in \script{J}_n^D(p, m, T)}(T_n^{\text{ad}}[j] + T_n^{\text{fr}}[j] )}\\
\hspace{-.1in} && + \expect{r_n[t_D] (t_D - \tilde{t}_n)}\\
\hspace{-.1in}&=&  \sum_{p, m, T}G_n(p, m)\expect{|\script{J}_n^D(p, m, T)|} \\
\hspace{-.1in} && + \expect{r_n[t_D] (t_D - \tilde{t}_n)}. 
\end{eqnarray*}
In this case, we have: 
\begin{eqnarray} 
\hspace{-.3in}&& \lim_{D\rightarrow\infty}\frac{\sum_n\sum_{d=0}^{D-1}\expect{\delta_dr_n[t_d]}}{\sum_{d=0}^{D-1}\expect{\delta_d}} \nonumber\\
\hspace{-.3in}&=&  \lim_{D\rightarrow\infty} \sum_n\sum_{p, m, T} G_n(p, m)\frac{\expect{|\script{J}_n^D(p, m, T)|}}{\sum_{d=0}^{D-1}\expect{\delta_d}}\label{eq:exp2actual}\\
\hspace{-.1in} && +  \lim_{D\rightarrow\infty}\sum_n\frac{\expect{r_n[t_D] (t_D - \tilde{t}_n)}}{\sum_{d=0}^{D-1}\expect{\delta_d}}\nonumber
\end{eqnarray}
Since $(t_D - \tilde{t}_n)\leq T_n^{\max}$, $r_n[t_D]\leq \frac{G^{\max}}{T_n^{\min}}$, and $\lim_{D\rightarrow\infty}\sum_{d=0}^{D-1}\expect{\delta_d}=\infty$, we see that the last term vanishes. 
Moreover, it can be shown that for all $S_n$: 
\begin{eqnarray*}
\hspace{-.3in}&&\lim_{D\rightarrow\infty}\frac{\expect{|\script{J}_n^D(p, m, T)|}}{\sum_{d=0}^{D-1}\expect{\delta_d}} \\
\hspace{-.3in}&=&\lim_{D\rightarrow\infty}\frac{\expect{|\script{J}_n^D(p, m, T)|}}{\sum_{p, m, T}\expect{|\script{J}_n^D(p, m, T)|}[F_n(p, m)+T]} \\
\hspace{-.3in}&=&  \frac{ z_n^{p, m, T} }{\sum_{p, m, T}z^{p, m, T}[F_n(p, m)+T]}. 
\end{eqnarray*}
Using this in  (\ref{eq:exp2actual}), we have: 
\begin{eqnarray} 
\hspace{-.3in}&& \lim_{D\rightarrow\infty}\frac{\sum_n\sum_{d=0}^{D-1}\expect{\delta_dr_n[t_d]}}{\sum_{d=0}^{D-1}\expect{\delta_d}} \nonumber\\
\hspace{-.3in}&&\qquad\qquad =   \sum_n\frac{\sum_{p, m, T}z_n^{p, m, T}G_n(p, m)}{\sum_{p, m, T}z_n^{p, m, T}[F_n(p, m) + T]}. 
\end{eqnarray}
This and (\ref{eq:exp2actual-0}) show that the average revenue satisfies: 
\begin{eqnarray}
\hspace{-.4in}&&\text{Profit}^{\textsf{AI}}_{\text{av}}\geq\text{Profit}^*_{\text{av}} -\frac{C_1}{V} - \frac{C_0}{V T^{\min}}.  
\end{eqnarray}

(Budget) We now prove that the average budget constraint (\ref{eq:budget-cond}) is ensured under the \textsf{AI} algorithm. 

We first prove (\ref{eq:q-bound}). To see that (\ref{eq:q-bound}) holds, note that $c_{\max}$ is the maximum change of $Q(t_d)$ during any interval $[t_d, t_{d+1})$. Thus, if $Q(t_d)<V\nu$ at time $k$, then we must have $Q(t_{d+1})\leq V\nu+c_{\max}$. On the other hand, suppose now $Q(t_d)>V\nu$. Without loss of generality, assume that $Q(t_d)$ is the first time $Q(t_d)>V\nu$ occurs. One see that $Q(t_{d})\leq V\nu+c_{\max}$. Also, we see that from now on until $Q(t_d)$ drops below $V\nu$ again, the advertiser will choose $p_n=0$ for any $S_n$ that requires update during this interval. This is because once $Q(t_d)>V\nu$, choosing any $p_n>0$ will result in $\Psi_n(t_d)<0$. This implies that $a_n[t_d]=0$ for all $S_n\in\script{S}(t_d)$ during this time.  Hence, the only possible increment of $Q(t_d)$ will be due to the sites that are in the middle of their frames. 
$Q(t_d)$ will also start to decrease once all the active sites complete their current frames. Now it can be seen that this increasing interval will last for at most $T^{\max}$ long, and for each active site, $a_n[t_d]\leq\frac{p^{\max}_n}{T_n^{\min}}$. Thus, the maximum possible increment after $Q(t_d)\geq V\nu$ will be $c_{\max}$. Combining this fact with $Q(t_{d})\leq V\nu+c_{\max}$, we conclude that $Q(t_{d})\leq V\nu+2c_{\max}$ for all $t_d$.

%



This shows that the deficit queue with $A[t_d]$ and $\mu[t_d]$ being the arrival and service rates is bounded. Using a similar argument as in the revenue case, one can show that:  
\begin{eqnarray} 
\hspace{-.4in}&& \sum_n\frac{\sum_{p, m, T}z_n^{p, m, T}p}{\sum_{p, m, T}z_n^{p, m, T}[F_n(p, m) + T]}\leq B_{\text{av}}. \nonumber
\end{eqnarray}
Using SLLN again, one sees  that the LHS equals the average  expenditure with probability $1$. Hence, (\ref{eq:budget-cond}) is satisfied.  
\end{proof}

\section*{Appendix D -- Proof of Theorem \ref{theorem:error}}
We prove Theorem \ref{theorem:error} here. 
\begin{proof}
(Revenue) First, we prove the revenue performance. 
Note that the $Q(t_d)$ value is now updated with $\hat{G}_n(p, m)$ and $\hat{F}_n(p, m)$. For convenience, we denote this queueing process by $\hat{Q}(t_d)$, and denote the effective revenue generation rate by $\hat{r}_n[t_d]$. 
Then, using the queueing dynamic equation (\ref{eq:deficit-queue}), we can similarly obtain: 
\begin{eqnarray}
\hspace{-.4in}&&\Delta(t_d) -V \delta_d\sum_n\hat{r}_n[t_d]\nonumber\\
\hspace{-.4in}&&\qquad\leq C_2 - \hat{Q}(t_d)\delta_dB_{\text{av}} \nonumber \\
\hspace{-.4in}&&\qquad\qquad\,\,\,- \delta_d\sum_n \frac{V\hat{G}_{n}[k_n(t_d)] - \hat{Q}(t_d)p_n[k_n(t_d)]}{\hat{F}_{n}[k_n(t_d)]+T^{\text{fr}}_{n}[k_n(t_d)]} \nonumber\\
\hspace{-.4in}&&\qquad\leq C_2 - \hat{Q}(t_d)\delta_dB_{\text{av}} - \delta_d\sum_n\hat{\Psi}_n(t_d). \label{eq:drift-ineq-err}
\end{eqnarray} 
Here $C_2\triangleq  \frac{1}{2}((1+\rho_F)T^{\max})^2\big[(\sum_n\frac{p_n^{\max}}{(1-\rho_F)T_n^{\min}})^2+B^2_{\text{av}}\big]$ and $\hat{\Psi}_n(t_d)$ is $\Psi_n(t_d)$ with $\hat{G}_n(p, m)$ and $\hat{F}_n(p, m)$.  

Similar to the proof of Theorem \ref{theorem:ai-per}, we will show that the term $\sum_n\hat{\Psi}_n(t_d)$ is roughly maximized under \textsf{AI} with estimation errors. To do so, note that the advertiser will now choose $p_n, m_n, T_n^{\text{fr}}$ by solving (\ref{eq:ai-obj}) with  $\hat{G}_n$, $\hat{F}_n$, and  $\hat{Q}(t_d)$. 

Consider any $S_n$ and $t_d$. Let $\tilde{t}_n$ be the last time when the advertiser updates $S_n$'s configuration before $t_d$.  Denote the actions chosen at time $\tilde{t}_n$ by \textsf{AI} with $\hat{G}_n$ and $\hat{F}_n$ by $\hat{p}_n$, $\hat{m}_n$, and $\hat{T}_n^{\text{fr}}$, and denote $p^*_n$, $m^*_n$, and $T_n^{\text{fr}*}$ any alternative feasible actions. 
%
Also, denote $\sigma_{gn}(p, m)=\hat{G}_n({p}, {m})-G_n({p}, {m})$ and $\sigma_{fn}(p, m)=\hat{F}_n({p}, {m})- F_n({p}, {m})$ the estimation errors of the functions. Then, we have: 
\begin{eqnarray}
\hspace{-.3in}&&\hat{\Psi}_n(t_d) =\frac{V[G_{n}(\hat{p}_n, \hat{m}_n)+\sigma_{gn}(\hat{p}, \hat{m})]- \hat{Q}(t_d)\hat{p}_n}{[F_{n}(\hat{p}_n, \hat{m}_n) +\sigma_{fn}(\hat{p}_n, \hat{m}_n)]+\hat{T}^{\text{fr}}_n}  \nonumber\\
\hspace{-.3in}&&= \frac{V[G_{n}(\hat{p}_n, \hat{m}_n)+\sigma_{gn}(\hat{p}, \hat{m})]- \hat{Q}(\tilde{t}_n)\hat{p}_n}{[F_{n}(\hat{p}_n, \hat{m}_n) +\sigma_{fn}(\hat{p}_n, \hat{m}_n)]+\hat{T}^{\text{fr}}_n}  \nonumber\\
\hspace{-.3in}&&\qquad\qquad\qquad + \frac{(\hat{Q}(\tilde{t}_n) - \hat{Q}(t_d))\hat{p}_n}{[F_{n}(\hat{p}_n, \hat{m}_n) +\sigma_{fn}(\hat{p}_n, \hat{m}_n)]+\hat{T}^{\text{fr}}_n} \nonumber\\
\hspace{-.3in}&&\stackrel{(a)}{\geq}  \frac{VG_{n}(p^*_n, m^*_n)  - \hat{Q}(t_d)p^*_n}{[F_{n}(p^*_n, m^*_n) +\sigma_{fn}(p^*_n, m^*_n)] +T^{\text{fr}*}_n}\nonumber\\
\hspace{-.3in}&&+  \frac{V\sigma_{gn}(p^*_n, m^*_n)  + ( \hat{Q}(t_d) -  \hat{Q}(\tilde{t}_n))p^*_n }{[F_{n}(p^*_n, m^*_n) +\sigma_{fn}(p^*_n, m^*_n) ] + T^{\text{fr}*}_n} - \frac{T_{\max}\hat{c}_{\max} p^{\max}}{(1-\rho_F)T^{\min}}\nonumber\\
\hspace{-.3in}&&\geq \Psi_n^*(t_d) \frac{F_{n}(p^*_n, m^*_n) + T^{\text{fr}*}_n}{[F_{n}(p^*_n, m^*_n) +\sigma_{fn}(p^*_n, m^*_n) ] + T^{\text{fr}*}_n} \nonumber \\
\hspace{-.3in}&&\qquad\qquad\qquad-  \frac{V \rho_GG^{\max} +2T^{\max}\hat{c}_{\max}p^{\max}  }{ (1-\rho_F) T^{\min} } \nonumber\\
\hspace{-.3in}&&\stackrel{(b)}{\geq} \frac{1}{1+\rho_F}\Psi_n^*(t_d) - \frac{V\rho_GG^{\max}+2T^{\max}\hat{c}_{\max}p^{\max}}{(1-\rho_F)T^{\min}}. \label{eq:error-1}
\end{eqnarray}
Here in inequality (a), we have used:  
\[\hat{c}_{\max}=(1+\rho_F)T^{\max}\max[\sum_n\frac{p_n^{\max}}{(1-\rho_F)T_n^{\min}},  B_{\text{av}}], \]
and the fact that $|\hat{Q}((t_d) -  \hat{Q}((t^*_n)|\leq T^{\max} \hat{c}_{\max}$. 
$\Psi_n^*(t_d)$ is defined as: 
\begin{eqnarray}
\hspace{-.3in}&&\Psi_n^*(t_d)=  \frac{VG_{n}(p^*_n, m^*_n)  -  \hat{Q}(t_d)p^*_n}{ F_{n}(p^*_n, m^*_n)  +T^{\text{fr}*}_n}.\label{eq:psi-opt}
\end{eqnarray}
In inequality (b), we have used the fact that $\Psi_n^*(t_d)\geq0$. 
Plug (\ref{eq:psi-opt}) and (\ref{eq:error-1}) back into (\ref{eq:drift-ineq-err}), we obtain: 
\begin{eqnarray}
\hspace{-.35in}&&\Delta(t_d) -V \delta_d\sum_n\hat{r}_n[t_d] \leq C_2+\delta_dC_3 \label{eq:error-3}\\
\hspace{-.35in}&&- \hat{Q}(t_d)\delta_dB_{\text{av}} - \delta_d \frac{1}{1+\rho_F}\sum_n\Psi^*_n(t_d)   + \frac{\delta_dVN\rho_GG^{\max}}{(1-\rho_F)T^{\min}}. \nonumber
\end{eqnarray} 
Here $C_3\triangleq \frac{2NT^{\max}\hat{c}_{\max}p^{\max}}{(1-\rho_F)T^{\min}}=\Theta(1)$. 

Now plug the stationary and randomized policy $\Pi_S$ into (\ref{eq:error-3}), we get:  
\begin{eqnarray}
\hspace{-.3in}&&\Delta(t_d) -V \delta_d\sum_n\hat{r}_n[t_d]\nonumber\\
\hspace{-.3in}&&\leq C_2+\delta_dC_3 \nonumber\\
\hspace{-.3in}&&\qquad- \hat{Q}(t_d)\delta_dB_{\text{av}} - \delta_d \frac{1}{1+\rho_F}\sum_n\Psi^{\Pi_S}_n(t_d)+ \frac{\delta_dVN \rho_GG^{\max}}{(1-\rho_F)T^{\min}} \nonumber\\
\hspace{-.3in}&&\leq C_2+\delta_dC_3 +  \frac{\delta_dV N\rho_GG^{\max}}{(1-\rho_F)T^{\min}}\nonumber\\
\hspace{-.3in}&& \qquad- \hat{Q}(t_d)\delta_d[B_{\text{av}}-\frac{1}{1+\rho_F} B_{\text{av}}] - \delta_d \frac{1}{1+\rho_F} V\text{Profit}^*_{\text{av}} \nonumber\\
\hspace{-.3in}&&\leq C_2+\delta_dC_3 +  \frac{\delta_dVN \rho_GG^{\max}}{(1-\rho_F)T^{\min}}- \delta_d \frac{1}{1+\rho_F} V\text{Profit}^*_{\text{av}}. \nonumber
\end{eqnarray} 
Using an argument similar to the one used in the proof of Theorem \ref{theorem:ai-per}, one can show that: 
\begin{eqnarray*} 
\hspace{-.3in}&& \lim_{D\rightarrow\infty}\frac{\sum_n\sum_{d=0}^{D-1}\expect{\delta_d\hat{r}_n[t_d]}}{\sum_{d=0}^{D-1}\expect{\delta_d}} \nonumber\\
\hspace{-.3in}&=&\sum_n\frac{\sum_{m, p, T}\hat{z}_n^{p, m, T}\hat{G}_n(p, m)}{\sum_{p, m, T}\hat{z}_n^{p, m, T}[\hat{F}_n(p, m)+T]}\\
\hspace{-.3in}&\geq& \frac{1}{1+\rho_F}\text{Profit}_{\text{av}}^* - \frac{C_2+T^{\min}C_3}{VT^{\min}}-\frac{N\rho_GG^{\max}}{(1-\rho_F)T^{\min}}. 
\end{eqnarray*}
Therefore, the actual average profit satisfies: 
\begin{eqnarray*} 
\hspace{-.2in} \text{Profit}_{\text{av}}^{\textsf{AI}} &=&\sum_n\frac{\sum_{m, p, T}\hat{z}_n^{p, m, T}G_n(p, m)}{\sum_{p, m, T}\hat{z}_n^{p, m, T}[F_n(p, m)+T]}\\
\hspace{-.2in} &\stackrel{(a)}{\geq}&\sum_n\frac{(1-\rho_F)\sum_{m, p, T}\hat{z}_n^{p, m, T}\hat{G}_n(p, m)}{(1+\rho_G)\sum_{p, m, T}\hat{z}_n^{p, m, T}[\hat{F}_n(p, m)+T]}\\
\hspace{-.2in}&\geq& \frac{1-\rho_F}{(1+\rho_F)(1+\rho_G)}\text{Profit}_{\text{av}}^* \\
\hspace{-.2in}&& - \frac{(C_2+T^{\min}C_3)(1-\rho_F)}{VT^{\min}(1+\rho_G)}-\frac{N\rho_GG^{\max}}{(1+\rho_G)T^{\min}}. 
\end{eqnarray*}
Here in  inequality (a) we have used the facts that $G_n(p, m)\geq \frac{\hat{G}_n(p, m)}{1+\rho_G}$ and $F_n(p, m)\leq\frac{\hat{F}_n(p, m)}{1-\rho_F}$ for all $n$ and all $p, m$.

(Budget) We now prove the budget performance. Using a similar argument as in the proof of Theorem \ref{theorem:ai-per}, we see that $\hat{Q}(t_d)$ is bounded by $(1+\rho_G)V\nu+2\hat{c}_{\max}$. However, different from Theorem  \ref{theorem:ai-per}, this does not imply that the actual expenditure does not violate the budget constraint $B_{\text{av}}$. This is so because the estimation errors  make $\hat{Q}(t_d)$ an approximation of the actual expected  budget deficit. 

To prove our result, we  denote $\hat{z}_n^{p, m, T}$ the fraction of time the triple $(p, m, T)$ is used under the \textsf{AI} algorithm with imperfect estimations $\hat{G}_n(\cdot, \cdot)$ and $\hat{F}_n(\cdot, \cdot)$. Then, since $\hat{Q}(t_d)$ is stable, we have: 
\begin{eqnarray}
\sum_n\frac{ \sum_{p, m, T}\hat{z}^{p, m, T}_n p  }{ \sum_{p, m, T}\hat{z}^{p, m, T}_n [\hat{F}_n(p, m) + T] }\leq B_{\text{av}}. 
\end{eqnarray}
Now using the fact that: 
\begin{eqnarray}
\hspace{-.3in}&&\sum_n\frac{ \sum_{p, m, T}\hat{z}^{p, m, T}_n p  }{ \sum_{p, m, T}\hat{z}^{p, m, T}_n [(1+\rho_F)F_n(p, m) + T] } \nonumber\\
\hspace{-.3in}&&\qquad\qquad\leq\sum_n\frac{ \sum_{p, m, T}\hat{z}^{p, m, T}_n p  }{ \sum_{p, m, T}\hat{z}^{p, m, T}_n [\hat{F}_n(p, m) + T] }\nonumber, 
\end{eqnarray}
we conclude that: 
\begin{eqnarray}
\sum_n\frac{ \sum_{p, m, T}\hat{z}^{p, m, T}_n p  }{ \sum_{p, m, T}\hat{z}^{p, m, T}_n [F_n(p, m) + T] } \leq (1+\rho_F)B_{\text{av}}. \label{eq:error-exp-bound}
\end{eqnarray}
Using a similar argument as in the proof of Theorem \ref{theorem:ai-per}, it can be shown that the LHS of (\ref{eq:error-exp-bound}) corresponds to the actual average advertising expenditure. This completes the proof of the theorem. 
\end{proof}

\vspace{.05in}

\bibliographystyle{unsrt}
\bibliography{mybib}

\end{document}